\documentclass[11pt,letterpaper]{article}
\usepackage{fullpage}
\usepackage[top=2cm, bottom=4.5cm, left=2.5cm, right=2.5cm]{geometry}
\usepackage{amsmath,amsthm,amsfonts,amssymb,amscd}
\usepackage{lastpage}
\usepackage{enumerate}
\usepackage{fancyhdr}
\usepackage{mathrsfs}
\usepackage{xcolor}
\usepackage{graphicx}
\usepackage{listings}
\usepackage[colorlinks=true, urlcolor=blue, linkcolor=blue, citecolor=magenta]{hyperref}
\usepackage{cite}
\usepackage{multirow}
\usepackage{comment}
\usepackage{mathpazo}

\hypersetup{%
  colorlinks=true,
  linkcolor=red,
  linkbordercolor={0 0 1}
}

\lstdefinestyle{Matlab}{
    language        = matlab,
    frame           = lines, 
    basicstyle      = \footnotesize,
    keywordstyle    = \color{blue},
    stringstyle     = \color{green},
    commentstyle    = \color{red}\ttfamily
}
\newtheorem{theorem}{Theorem}
\newtheorem{definition}[theorem]{Definition}

\newtheorem{question}{Question}

\newtheorem{lemma}[theorem]{Lemma}
\newtheorem{hypothesis}[theorem]{Hypothesis}
\newtheorem{corollary}[theorem]{Corollary}

\newcommand{\setcover}{\textsc{SetCover}}
\newcommand{\ksetcover}{\textsc{$k$-SetCover}}
\newcommand{\wtwo}{\mathsf{W[2]}}
\newcommand{\wone}{\mathsf{W[1]}}
\newcommand{\fpt}{\mathsf{FPT}}

\begin{document}

\title{Constant Approximating Parameterized $k$-SetCover is W[2]-hard}
\author{
Bingkai Lin \thanks{Nanjing University. Email: \texttt{lin@nju.edu.cn}}
\and
Xuandi Ren\thanks{University of California, Berkeley. Email: \texttt{xuandi\_ren@berkeley.edu}}
\and
Yican Sun\thanks{School of Computer Science, Peking University. Email: \texttt{sycpku@pku.edu.cn}}
\and 
Xiuhan Wang\thanks{Tsinghua University. Email: \texttt{wangxh19@mails.tsinghua.edu.cn}}
}

\maketitle

\begin{abstract}
    In this paper, we prove that it is $\wtwo$-hard to approximate \ksetcover{} within any constant ratio. Our proof is built upon the recently developed  threshold graph composition technique. We propose a strong notion of  threshold graphs and use a new composition method to prove this result. Our technique could also be applied to rule out polynomial time $o\left(\frac{\log n}{\log \log n}\right)$ ratio approximation algorithms for the non-parameterized \ksetcover{} problem with $k$ as small as $O\left(\frac{\log n}{\log \log n}\right)^3$, assuming $\wone \ne \fpt$. 
    We highlight that our proof does not depend on the well-known PCP theorem, and only involves simple combinatorial objects. 
\end{abstract}

\section{Introduction}
\label{sec:intro}
 
In the \ksetcover{} problem, we are given a bipartite graph $G = (S \,\dot\cup\, U, E)$ and an integer $k$, and the goal is to decide whether there exist $k$ vertices in $S$ such that each node in the \emph{universe set} $U$ is a neighbor of one of the $k$ vertices.  In the classic complexity regime, this problem is known to be {\sf NP}-complete \cite{kar72}. Thus assuming $\mathsf{P}\neq \mathsf{NP}$, no algorithm can solve \ksetcover{} in polynomial time. To circumvent 
this intractability, many previous works have focused on efficient approximation algorithms for \ksetcover{}. It admits a simple greedy algorithm with approximation ratio $(\ln n - \ln \ln n + \Theta(1))$~ \cite{cha79,joh74,Lov75,sla97,ste74}. On the opposite side, the hardness of approximation of \ksetcover{} has also been intensively studied~ \cite{alon2006algorithmic,dinste14,feige1998threshold,lunyan94,razsaf97}. The state-of-the-art result by Dinur and Stenur~\cite{dinste14} shows that approximating  \ksetcover{} within an $(1-\varepsilon)\cdot \ln n$ factor for any $\varepsilon>0$ is {\sf NP}-hard.

To obtain a more fine-grained comprehension of  {\sf NP}-hard problems, it is natural to consider \emph{parameterization}. In the parameterized complexity regime, people wonder whether there is any  $f(k)\cdot |G|^{O(1)}$ time algorithm (FPT algorithm) that can solve the  \ksetcover{} problem, where $f$ could be any computable function (e.g., $k!$ or $k^{k^k}$). 
As was shown in \cite{DowneyF95}, \ksetcover{} is the canonical {\sf W}[2]-complete problem, which means that, unless $\mathsf{W[2]} = \mathsf{FPT}$, \ksetcover{} does not admit any FPT algorithm. In parallel with the classical complexity regime, it is natural to ask:
\begin{center}
    \textit{Assuming $\mathsf{W[2]}\neq \mathsf{FPT}$, does {\sc \ksetcover{}} admit good FPT approximation algorithms?}
\end{center}
Note that since \ksetcover{} is a $\mathsf{W[2]}$-complete problem, the question above is (almost) equivalent to asking whether there is  an FPT self-reduction from exact \ksetcover{} to its gap version\footnote{For any $c>1$, on input a \ksetcover{} instance $(G,k)$, the goal of $c$-gap \ksetcover{} is to distinguish between the cases where $G$ has $k$-size solution and the cases where $G$ has no $ck$-size solution.}. 

In recent years, the parameterized inapproximability of \ksetcover{} has been established under various assumptions stronger than $\wtwo \ne \fpt$. Chen and Lin~ \cite{CL19} proved that it is $\mathsf{W[1]}$-hard to approximate \ksetcover{} within any constant ratio. 
Chalermsook et al. ~\cite{CCK+17} showed that assuming {\sf Gap-ETH}, \ksetcover{} cannot be approximated within a $(\log n)^{O(1/k)}$ factor in $n^{o(k)}$ time.  Using the \emph{Distributed PCP framework} \cite{abboud2017distributed}, Karthik, Laekhanukit and  Manurangsi \cite{KLM19} ruled out $(\log n)^{1/\text{poly}(k)}$  ratio  $n^{k-\varepsilon}$-time approximation algorithms for \ksetcover{} under {\sf SETH}. They also proved similar inapproximability results under  {\sf $k$-SUM Hypothesis}, {\sf ETH} and $\mathsf{W[1]} \neq \mathsf{FPT}$, respectively.  The hardness of approximation factors of these results were improved to $(\log n/\log\log n)^{1/k}$ in~\cite{Lin19} by combining \ksetcover{} instances with some gap-gadgets. Karthik and Navon~\cite{KL21} gave a simple construction of these gap-gadgets based on error correcting codes and named this technique \emph{Threshold Graph Composition}.

In contrast to the success of proving inapproximability of \ksetcover{} under stronger assumptions, 
the  FPT inapproximability of \ksetcover{}  under $\mathsf{W[2]} \neq \mathsf{FPT}$ remained completely open \cite{KLM19,KL21,feldmann2020survey}. Below we  summarize the difficulties encountered in the previous approaches.


\medskip\noindent\textbf{Barriers of previous approaches}. 
First, neither of the two pioneering works \cite{CCK+17, CL19} are applicable to this setting. In detail, the approach proposed in \cite{CL19} requires the reduction's starting point to admit a product structure (e.g. {\sc $k$-Clique}) in order to perform gap amplification. However, under the assumption $\mathsf{W[2]}\ne \mathsf{FPT}$, the starting point is the \ksetcover{} itself, whose product structure is notoriously hard to understand \cite{VizingConjecture}. On the other hand, the method of \cite{CCK+17} highly depends on the inherent gap in the {\sf Gap-ETH} assumption, thus could not be adapted to our setting.

Second, there are indeed some other advanced techniques, namely, the Distributed PCP Framework and the Threshold Graph Composition, which have been successfully used to prove the inapproximability of \ksetcover{} under $\mathsf{W[1]}\ne \mathsf{FPT}$, $\mathsf{ETH}$, and $\mathsf{SETH}$. However, under a weaker assumption $\mathsf{W[2]}\ne \mathsf{FPT}$, even constant FPT inapproximability of \ksetcover{} is difficult to reach using these two techniques. Generally speaking, in order to create a constant gap for \ksetcover{} under $\mathsf{W[2]}\ne \mathsf{FPT}$, both of these methods need to first establish the $\mathsf{W[2]}$-hardness of \ksetcover{} with a small universe set. Nevertheless, establishing such a hardness result is highly non-trivial, and we believe it to be unrealizable\cite{BA17}. 

Below, we illustrate their technique barriers in detail, respectively.

\begin{itemize}
    \item The work \cite{Lin19} proposes a gap-producing self-reduction for \ksetcover{}, i.e., it constructs a reduction from \ksetcover{} to $c$-gap \ksetcover{}. However, to create the constant gap $c$, the running time of this reduction will be $|U|^{(ck)^k}$, where $|U|$ represents the universe size for the input \ksetcover{} instance. Hence, to apply this method under our setting $\mathsf{W[2]}\ne \mathsf{FPT}$, we need to first prove that \ksetcover{} remains $\mathsf{W}[2]$-hard even when $|U|=n^{O(1/{(ck)}^k)}$.
    \item The approach of~\cite{KLM19} is more complicated. It starts with a $(1-1/k)$-gap $k$-{\sc MaxCover}\footnote{In the $k$-{\sc MaxCover} problem, we are given a bipartite graph $I=(V,U,E)$ with $|V|=k$ and $|U|=q$, two alphabet $\Sigma_V$ and $\Sigma_U$ and constraints $C_{e}\subseteq \Sigma_V\times\Sigma_U$ for each $e\in E$. The goal is to find assignments $\sigma_V : V\to\Sigma_V$ and $\sigma_U : U\to\Sigma_U$ to maximize the number of $u\in U$ whose edges's constraints are all satisfied (Such vertex $u$ are called \emph{covered} by assignments $\sigma_V$ and $\sigma_U$). For $\delta\in(0,1)$, the $\delta$-gap $k$-{\sc MaxCover}  is to distinguish between the cases where all vertices in $U$ can be covered and the cases where at most $\delta$-fraction of $U$ can be covered.} instance which has $q$ right variables  and  right alphabet $\Sigma_U$. It then uses the product method to amplify its gap to $(1-1/k)^t$ with the price of increasing $q$ to $q^t$ and $\Sigma_U$ to $\Sigma_U^t$. Finally it transfers a $(1-1/k)^t$-gap $k$-Max-Cover instance with alphabet $\Sigma_U^t$ and variable $q^t$ to a  $(1-1/k)^{-t/k}$-gap \ksetcover{} instance with size at least $q^t\cdot k^{|\Sigma_U|^t}$. It is not hard to see that to obtain a constant gap for \ksetcover{}, $t$ should be at least $\Omega(k^2)$, which means that $q$ should be at most $n^{O(1/k^2)}$. However, as far as we known, there is no FPT-reduction from \ksetcover{} to $(1-1/k)$-gap $k$-{\sc MaxCover} with $q$ at most $n^{O(1/k^2)}$.
    We note that there is a simple reduction from \ksetcover{} to $k$-{\sc MaxCover} with $q=|U|$ and $\Sigma_U=[k]$. Such a reduction is far from satisfactory because the $k$-{\sc MaxCover} instance it produces does not have a $(1-1/k)$-gap and it requires the $\mathsf{W[2]}$-hardness of \ksetcover{} with $|U|=n^{O(1/k^2)}$.
\end{itemize}

Note that establishing such a hardness result is equivalent to making an FPT self-reduction that ``compresses'' the \ksetcover{} problem, i.e., the reduction starts with a \ksetcover{} instance with universe size $n$ and results in another equivalent instance of \ksetcover{} with significantly smaller universe size $n^{o(1)}$.  
Unfortunately,  there is no known FPT-reduction that achieves this, and we list two partial results below, which suggest that such a reduction might be unrealizable.
\begin{itemize}
\item First, we  prove that (See Lemma \ref{lem:setcoverw1} in the Appendix), unless $\mathsf{W[2]} = \mathsf{W[1]}$, there is no  FPT self-reduction for \ksetcover{} that can compress the size of universe set to $f(k)\cdot \log n$ for any computable function $f$. 
\item In addition, research from the lower bound of kernelization~\cite{BA17} shows that: unless $\mathsf{NP}\subseteq \mathsf{coNP/poly}$ and the polynomial hierarchy $\mathsf{PH}$ collapses, for every $\varepsilon>0$, there is no polynomial self-reduction for non-parameterized \ksetcover{} that compress the  size of instance to $n^{2-\varepsilon}$.
This also suggests that the corresponding parameterized version of \ksetcover{} may also be hard to compress. 
\end{itemize}

\medskip\noindent\textbf{Our results}. We bypass these technique barriers by coming up with a strong version of the threshold graph and a new composition method that can create a constant gap for \ksetcover{} instance with a large universe set. Using this technique, we give the first inapproximability result for \ksetcover{} under $\mathsf{W[2]} \neq \mathsf{FPT}$. In summary, the main contribution of our paper is:
\begin{theorem}
\label{w2hardness}
    Assuming  $\mathsf{W[2]} \neq \mathsf{FPT}$, there is no  FPT algorithm which can approximate {\sc\ksetcover{}} within any constant ratio.
\end{theorem}

As a further application, our technique can also be used to rule out polynomial time $o\left(\frac{\log n}{\log\log n}\right)$ ratio approximation algorithms for  \ksetcover{} with $k$  as small as $O(\log^3 n)$ under  $\mathsf{W[1]} \neq \mathsf{FPT}$. 
\begin{theorem}
\label{w1hardness}
Assuming $\mathsf{W[1]} \neq \mathsf{FPT}$, there is no polynomial time algorithm which can approximate \ksetcover{} within $o\left(\frac{\log n}{\log\log n}\right)$ ratio, even if $k$ is as small as $O\left(\frac{\log n}{\log \log n}\right)^3$.
\end{theorem}


We emphasize that our Theorem \ref{w1hardness} is non-trivial since it overcomes limitations of previous works in two aspects.



\begin{itemize}
    \item \emph{Limitations of results in parameterized complexity.} First, all the previous results~\cite{CL19,CCK+17,KLM19,Lin19} on the parameterized inapproximability of \ksetcover{} only ruled out \\ $o\left((\log n)^{1/k}\right)$ approximation ratios, which was further pointed out as a barrier\cite{KL21}. However, the folklore polynomial time greedy algorithm could approximate \ksetcover{} with approximation factor  $\Theta(\log n)$. There is a huge gap between them. Our Theorem \ref{w1hardness} is the first one to obtain inapproximability result for \ksetcover{} with a small solution size that can bypass the $(\log n)^{1/k}$ barrier. 
    \item \emph{Limitations of the PCP theorem.} Second, the major technique used to rule out polynomial approximation algorithm for \ksetcover{} is the well-celebrated PCP theorem\cite{feige1998threshold, alon2006algorithmic}. However, as pointed out in \cite{CL19},  results of this type are unlikely provable using the classic PCP machinery since the PCP theorem always produces \ksetcover{} instances with a large solution size $\Omega(n)$. However, our Theorem \ref{w1hardness} could rule out polynoimal time approximation with a small solution size $O\left(\frac{\log n}{\log \log n}\right)^3$. Furthermore, assuming the mild assumption $\mathsf{NP}\nsubseteq \mathrm{TIME}(2^{\mathrm{poly}\log n})$,  \ksetcover{} with solution size $\log^{O(1)} n$ is not even $\mathsf{NP}$-hard, thus we could not derive hardness for \ksetcover{} with $\mathrm{poly}\log n$ solution size from gap SAT (or the PCP theorem).
\end{itemize}


Finally, as a suggestion to future work, we remark that: if one could improve the polynomial time lower bound in Theorem~\ref{w1hardness} to $n^{k-o(1)}$ under $\mathsf{SETH}$, then  the ratio $o\left(\frac{\log n}{\log\log n}\right)$ is tight.
The reason is that: there is a simple algorithm with running time $n^{k-\Omega(1)}$ that approximates \ksetcover{} with approximation factor $O\left(\frac{\log |U|}{\log k}\right)$. Under our setting, we have that $k=\log^{\Theta(1)} n$, hence, the algorithm could approximate \ksetcover{} with factor $O\left(\frac{\log |U|}{\log k}\right)=O\left(\frac{\log n}{\log\log n}\right)$ in $n^{k-\Omega(1)}$ time. The algorithm simply modifies the well-known greedy algorithm for \ksetcover{}, and is illustrated detailedly as follows.
\begin{itemize}
    \item[]\textit{Details of the algorithm.} Fix some integer $2\le T< k$, our algorithm  repeatedly picks $k-T$ vertices into the final solution in  $O(n^{k-T+1})$ time, such that these vertices cover the largest fraction of remaining elements in the universe. Repeat the procedure above for $r$ times with $r$ to be determined, we could obtain a solution of size $r \cdot (k-T)$. We analyze the correctness and the approximation ratio of our algorithm below. 
    \item[]\textit{Analysis.} In each repetition, we will find $k-T$ vertices that cover at least $1-T/k$ fraction of remaining elements in the universe. Thus after $r$ repetitions, the number of the uncovered elements will be no more than $(T/k)^r\cdot |U|$. To ensure that we have covered all elements in the universe, we need to choose the number of repetitions $r = \frac{\log |U|}{\log k - \log T}$ by solving the inequality for uncovered elements $(T/k)^r\cdot |U|< 1$.
\end{itemize}
Note that proving tight lower bounds and upper bounds for approximation algorithms of \ksetcover{} for different choices of $k$ and approximation ratios is still a big challenge. We would like to mention that our proof might shed a new light on the study of this problem since it is an elementary proof, which circumvents the complex PCP machinery and only involves simple combinatorial objects such as error correcting codes.








\subsection{Our Techniques}
\label{subsec:ourtech}
Our main technique is the \emph{threshold graph composition}, which was first used in \cite{Lin18}, and has been applied to create gaps for many parameterized problems \cite{CL19, Lin19, BBE+21, KL21}. At a very high level,  this method first constructs a  graph with some threshold properties and then combines this graph with the input instance to produce a gap instance of the desired problem. 

In this paper, we introduce a strong variant of the threshold graph and propose a new way to compose this threshold graph with the original \ksetcover{ }instance.  
Below we firstly overview the proof in the work\cite{Lin19},  and then illustrate our proof for Theorem \ref{w2hardness}
 briefly. Theorem \ref{w1hardness} could be analogously proved by simply choosing  another combination of parameters.

\medskip

\noindent\textbf{The Original Proof in \cite{Lin19}.} The threshold graph appeared in~\cite{Lin19} are 
bipartite graphs $T=(A\dot\cup B,E_T)$  with the following properties:
\begin{description}
    \item[(i)] $A=A_1\dot\cup A_2\dot\cup\cdots\dot\cup A_k$.
    \item[(ii)] $B=B_1\dot\cup B_2\dot\cup\cdots\dot\cup B_m$.
    \item[(iii)] For any $a_1\in A_1,\ldots,a_k\in A_k$ and $i\in [m]$, $a_1,\ldots,a_k$ have a common neighbor in $B_i$.
    \item[(iv)] For any $X\subseteq A$ and $b_1\in B_1,\ldots,b_m\in B_m$, if every $b_i$ has $k+1$ neighbors in $X$, then $|X|>h$.
\end{description}

Given a set cover instance $\Gamma=(S \dot\cup U,E)$ and a threshold graph $T=(A \dot\cup B,E_T)$ with $A=A_1\dot\cup A_2,\ldots,A_k$, $B=B_1\dot\cup,\cdots,\dot\cup B_m$ and $|A_i|=|S|$, the reduction in~\cite{Lin19} treats each $A_i$ as a copy of $S$ and creates a set cover instance $\Gamma_i=(A,U^{B_i},E_i)$ for every $i\in[m]$ so that in order to cover $U^{B_i}$, one has to pick $\ell$ notes $\{a_1,\ldots,a_\ell\}$ from $A$ satisfying the following two conditions:
\begin{description}
    \item[(a)] $a_1,\ldots,a_\ell$ cover $U$ in the instance $\Gamma$,
    \item[(b)] $a_1,\ldots,a_\ell$ have a common neighbor in $B_i$ in the threshold graph $T$.
\end{description}
After that, it takes the union of all $\Gamma_i=(A,U^{B_i},E_i)$ for $i\in[m]$. It is not hard to see that if the input instance $\Gamma$ has a solution of size $k$, then there exist $a_1\in A_1,\ldots,a_k\in A_k$ which cover all vertex in $U$. By (a), (b) and the property (iii), these vertices $a_1,\ldots,a_k$ also cover  $U^{B_i}$ for all $i\in[m]$.

On the other hand, if the input instance $\Gamma$ has no solution of size $k$,
below we show that every solution $X$ of the output instance must have a size larger than $h$.
To cover the part $\Gamma_i$ in the output instance, any solution $X$ has to incorporate $k+1$ vertices in to cover $U$ to satisfy condition (a). Furthermore, these vertices share a common neighbor $b_i\in B_i$ to satisfy condition (b). This means that we could find a vertex $b_i\in B_i$ with $k+1$ neighbors in our solution $X$. Since this works for every $i\in [m]$, we conclude by the threshold property (iv) that $|X|> h$.

The reduction above fails to prove constant inapproximability of {\sc \ksetcover{}} under $\mathsf{W[2]} \neq \mathsf{FPT}$ because to get a constant hardness factor, one needs to set $h=c\cdot k$ for some constant $c$. Unfortunately, the current construction of threshold graph has $|B_i|=h^k$.  Thus the reduction has running time at least $|U|^{{h}^k}$, which is not FPT when $|U|=\Omega(n)$. Previous reductions~\cite{Lin19,KLM19} use the fact that under stronger hypotheses, {\sc \ksetcover{}} remains $\mathsf{W[1]}$-hard even when $|U|=k^{O(1)}\log n$. However, unless $\mathsf{W[2]}=\mathsf{W[1]}$, one cannot prove {\sc \ksetcover{}} remains $\mathsf{W[2]}$-hard when $|U|=k^{O(1)}\log n$ (See Lemma~\ref{lem:setcoverw1} in the Appendix). Jansen and Pieterse~\cite{BA17} proved that unless $\mathsf{NP}\subseteq \mathsf{coNP/poly}$, there is no $O(n^{2-\epsilon})$-size compression for general non-parameterized \ksetcover{}, which suggests it might not be possible to prove \ksetcover{} with small $|U|$ is $\mathsf{W[2]}$-hard. 

\medskip

\noindent\textbf{Our Reduction.} Below we illustrate our reduction at a high level. We firstly define a stronger version of the threshold graph by
replacing (iv) with a  property  reminiscent of the soundness condition of \emph{multi-assignment} PCP \cite{alon2006algorithmic}.
\begin{description}
    \item[(iv')] For any $X\subseteq A$ and $I\subseteq [m]$ with $|I|\ge\varepsilon m$, if for every $i\in I$,  $b_i$ has $k+1$ neighbors in $X$, then $|X|>h$.
\end{description}
It turns out that such strong threshold graphs can be constructed using Error Correcting Codes, which was proposed by Karthik  and  Navon \cite{KL21}. 

Then we define a new threshold graph composition which replaces $\Gamma_i$ by $\Gamma_i'=(A\cup B_i,U^{c}\times B_i^{c},E_i')$. By choosing the edge set $E_i'$ carefully, we can guarantee that in order to cover $U^{c}\times B_i^{c}$, one has to
\begin{description}
    \item[(a)] either pick $c+1$ vertices from $B_i$,
    \item[(b)] or pick $a_1,\ldots,a_\ell$ from $A$ and $b_i\in B_i$ such that
\begin{description}
    \item[(b.1)] $a_1,\ldots,a_\ell$ cover $U$ in the instance $\Gamma$,
    \item[(b.2)] $b_i$ is a common neighbor of $a_1,\ldots,a_\ell$  in the threshold graph.
\end{description}
\end{description}

Suppose the input instance $\Gamma$ has a solution of size $k$. Since each $A_i$ is a copy of the set $S$ in $\Gamma$, let $a_1\in A_1,\ldots,a_k\in A_k$ be the vertices that can cover $U$. By the property of threshold graph,  $a_1,\ldots,a_k$ have a common neighbor $b_i\in B_i$ for every $i\in[m]$. The vertices $a_1,\ldots,a_k$ and $b_1,\ldots,b_m$ together can cover $U^{c}\times B_i^{c}$ for all $i\in[m]$.

On the other hand, if $\Gamma$ has no $k$-size solution, then in order to cover every $U^{c}\times B_i^{c}$, one should pick either $(1-\varepsilon)(c+1)m$ vertices in $B$ or pick $X\subseteq A$ and $\epsilon m$ vertices from $B$ such that each vertex has $k+1$ neighbors in $X$. By the property (iv'), the later implies that $|X|>h$. Thus, either  $(1-\varepsilon)(c+1)m$ vertices in $B$  or $h$ vertices in $A$ must be picked in this case. To obtain a constant gap, we assign $m/k$ weight to every vertex in $A$ and let $h=k^2$, $m=k^5$. It is routine to check that the reduction is FPT.

We remark that $|B_i|$ can be sufficiently small in the threshold graph construction. Assuming $\mathsf{W[1]}\neq \mathsf{FPT}$, $|U|$ is $k^{O(1)}\log n$ and thus we can set $c=\frac{\log n}{k \log \log n}$. Since $k$ is arbitrarily small with respect to $n$, by choosing appropriate parameters, we can rule out $o\left(\frac{\log n}{\log \log n}\right)$ factor polynomial time algorithms for the non-parameterized \ksetcover{} problem.

\bigskip

\section{Preliminaries}
\label{sec:prelims}

In this section, we formally define the \ksetcover{} problem and  related hypothesis.

For every graph $G$, we use $V(G)$ and $E(G)$ to denote its vertex set and edge set. For every vertex $v\in V(G)$, let $ N(v)\subseteq V(G)$ be the set of neighbors of $v$ in $G$. For every vertex set $C\subseteq V(G)$, define $ N(C)=\bigcup_{v\in C} N(v)$. 

The (weighted variant of) \setcover{} is defined as follows. An instance $\Gamma$ consists of a bipartite graph $G=(S,U,E)$ and a weight function $w: S\to \mathbb{N}^+$. The goal is to find a set $C\subseteq S$ such that $ N(C) = U$ and the total weight $\sum_{s\in C} w(s)$ is minimal. We use ${\sf OPT}(\Gamma)$ to represent the minimum total weight. For every $\chi\in\mathbb{N}$, we say an instance $\Gamma$ is \emph{$\chi$-weighted} if the number of different weights is upper bounded by $\chi$, i.e. $|\{w(v) : v\in S\}|\le\chi$. The unweighted \setcover{} is equivalent to $1$-weighted \setcover{}.

Given a constant $\chi$, by duplicating every vertex according to its weight, any $\chi$-weighted \setcover{} instance $\Gamma$ can be reduced to an unweighted ($1$-weighted) \setcover{} instance $\Gamma'$ in time $\left(|\Gamma| \cdot \max\{w(s): s\in S\}\right)^{O(1)}$, while preserving the optimum. 

\begin{lemma}[Lemma 16 in \cite{CL19}]
\label{lem:remove-weights}
    For any constant $\chi$, there is a reduction which, given any $\chi$-weighted \setcover{} instance $\Gamma = (S, U, E, w)$, outputs an unweighted \setcover{} instance $\Gamma' = (S', U', E')$ in \\$O(\left(|\Gamma|\max\{w(s): s\in S\} \right)^{O(1)})$ time, such that for every $k<|\Gamma|$, ${\sf OPT}(\Gamma)\le k$ if and only if ${\sf OPT}(\Gamma')\le k$.
\end{lemma}

In parameterized complexity theory, we consider problems $L\subseteq\{0,1\}^*$ with a computable function $\kappa : \{0,1\}^*\to\mathbb{N}$ which returns a parameter $\kappa(x)\in\mathbb{N}$ for every input instance $x\in \{0,1\}^*$. Since we mainly focus on graph related problems, it is convenient to treat each input as a pair $x=(G,k)$ with $G$ a graph and $k$ an integer and let $\kappa(x)=\kappa(G,k)=k$.
A parameterized problem  $(L,\kappa)$ is \emph{fixed parameter tractable} (FPT) if it has an  algorithm  which for every input $x\in\{0,1\}^*$ decides if $x$ is a yes-instance of $L$ in $f(\kappa(x))\cdot |x|^{O(1)}$-time for some computable function $f$. An FPT-reduction from problem $(L,\kappa)$ to $(L',\kappa')$ is an algorithm $A$ which on every input $x$, outputs an instance $x'$ in $f(\kappa(x))\cdot |x|^{O(1)}$-time such that $x$ is a  yes-instance if and only if $x'$ is a yes-instance  and $\kappa'(x')=g(\kappa(x))$ for some computable function $f$ and $g$. 
In the parameterized version of \setcover{}, i.e., \ksetcover{}, we are given  an instance $\Gamma$ of unweighted  \setcover{} and an integer $k$ as its parameter(i.e. $\kappa(\Gamma,k)=k$). The goal is to decide if ${\sf OPT}(\Gamma)\le k$. 
The $c$-gap \ksetcover{} problem is to distinguish between the cases ${\sf OPT}(\Gamma) \le k$ and ${\sf OPT}(\Gamma)>ck$. Another fundamental parameterized problem is the  $k$-{\sc Clique} problem, whose goal is to decide whether an input graph $G$ contains a clique of size $k$.

At the end, we present the hypothesis $\wtwo\ne\fpt$ and $\wone \ne \fpt$ on which our results based. For simplicity, we omit the definition of W-Hierarchy and only use an equivalent form.

\begin{hypothesis}[$\wtwo\ne\fpt$]
\ksetcover{} cannot be solved in $f(k)\cdot n^{O(1)}$ time for any computable function $f$.
\end{hypothesis}

\begin{hypothesis}[$\wone\ne\fpt$]
 $k$-{\sc Clique} cannot be solved in $f(k)\cdot n^{O(1)}$ time for any computable function $f$.
\end{hypothesis}

Analogously  to the definition of  $\mathsf{NP}$-hardness, we say a problem is   $\mathsf{W[1]}$-hard or  $\mathsf{W[2]}$-hard if there is an FPT-reduction from $k$-{\sc Clique} or \ksetcover{} to it, respectively.
Similarly, a problem is  in  $\mathsf{W[1]}$ or  $\mathsf{W[2]}$ if there is an FPT-reduction from it to $k$-{\sc Clique} or \ksetcover{}, respectively. Note that $\wone\ne\fpt$ implies $\wtwo\ne\fpt$, because there is an FPT reduction which transforms a $k$-{\sc Clique} instance to an unweighted \ksetcover{} instance $\Gamma$ with small universe size $|U| = k^3\cdot \log |\Gamma|$  \cite{KLM19,Lin19}. Thus, in our paper, the starting point of our reduction for Theorem~\ref{w1hardness} is an unweighted \ksetcover{} with small universe size rather than a $k$-{\sc Clique} instance. Formally, the reduction \cite{KLM19,Lin19} is stated below.

\begin{lemma}
\label{lemma:smallu}
    There is a polynomial time reduction which, given an $n$-vertex graph $G = (V, E)$ and a parameter $k$, outputs a \ksetcover{} instance $\Gamma = (S, U, E)$ where $|U| = O(k^3\log n)$ and $|S| = |E| = O(n^2)$, such that:
    \begin{itemize}
        \item if $G$ contains a $k$-clique, then $\mathsf{OPT}(\Gamma)\le \binom{k}{2}$;
        \item if $G$ contains no $k$-clique, then $\mathsf{OPT}(\Gamma) > \binom{k}{2}$.
    \end{itemize}
\end{lemma}

\section{Strong Threshold Graphs}
\label{subsec:stgraph}

In this part, we introduce a combinatorial object called strong threshold graph, which plays an important role in our proof. Intuitively, we compose strong threshold graphs with the original \ksetcover{} instance to produce gaps.
Our threshold graph construction comes from \cite{KL21}. However, we take a new analysis to this construction and establish the stronger threshold property on this type of graphs. Since the construction involves error correcting codes. For the sake of self-containedness, we first put the definition of error correcting codes here. Then, we formally define the strong threshold property used in this work, and finally, we use a new analysis to the construction in \cite{KL21} and prove (See Theorem \ref{thm:graph}) that it is a strong threshold graph.

\begin{definition}[Error Correcting Codes]
	Let $\Sigma$ be a finite set, a subset $\mathcal C:\Sigma^r \to \Sigma^m$ is an error correcting code with message length $r$, block length $m$ and relative distance $\delta$ if for every $x,y \in \Sigma^r$, $\Delta(\mathcal C(x),\mathcal C(y))\ge \delta$.  We denote then $\Delta(\mathcal C)=\delta$. Here $\Delta(x,y)=\frac{1}{m} |\{i \in [m]: x_i \ne y_i\}|$.
\end{definition}

We sometimes abuse notations a little and treat an error correcting code as its image, i.e., $\mathcal C \subseteq \Sigma^m$ and $|\mathcal C|=\Sigma^r$.

Throughout our paper, we use Reed-Solomon Code (RS code). Let $\Sigma$ be a field and $r\le m\le|\Sigma|$. Fix $m$ elements $f_1,\ldots,f_m\in\Sigma$.  The  RS code $\mathcal C^{\rm RS}: \Sigma^r\to \Sigma^m$ is defined as
\[
\forall (a_1,\ldots,a_r)\in\Sigma^r, \mathcal C^{\rm RS}(a_1,a_2,\ldots,a_r)=(\sum_{i\in[r]}a_if_1^{i-1},\sum_{i\in[r]}a_if_2^{i-1},\ldots,\sum_{i\in[r]}a_if_m^{i-1}).
\]

\begin{lemma}[RS code \cite{RSCode}]
\label{lemma:RS}
Given input and output length $r, m$ and alphabet $\Sigma$ such that $|\Sigma|\ge m$, the RS code $\mathcal C^{\rm RS}: \Sigma^r\to \Sigma^m$ satisfies $\Delta(\mathcal C^{\rm RS}) \ge 1-\frac{r}{m}$. 
\end{lemma}

Below we describe the properties of threshold graphs we need. Note that in Section \ref{subsec:ourtech} we have compared our threshold property with those in previous work\cite{Lin19}.

\begin{definition}[$(n, k, t, m, h,\varepsilon)$-threshold graph]
\label{def:graph}
    Given $n,k,t,m,h\in \mathbb N^+$ and $\varepsilon\in(0,1)$, a bipartite graph $T = (A\dot\cup B, E)$ is an $(n,k,t,m,h,\varepsilon)$-threshold graph if it has the following properties:
    \begin{itemize}
        \item $A$ consists of $k$ disjoint parts $A = A_1\dot\cup A_2\dot\cup \ldots \dot\cup A_k$ with $|A_i|=n$ for all $i\in[k]$.
        \item $B$ consists of $m$ disjoint parts $B = B_1\dot\cup B_2\dot\cup \ldots \dot\cup B_m$ with $|B_i|=t$ for all $i\in [m]$.
		\item For any $a_{1}\in A_{1},a_{2}\in A_{2}\ldots,a_{k}\in A_{k}$ and every $j \in [m]$, there is a vertex $b \in B_j$ which is a common neighbor of $\{a_{1},\ldots,a_{k}\}$.
		\item For any $X\subseteq A$ and $b_1\in B_1,\ldots,b_m\in B_m$, if there are $\varepsilon m$ indices $j$ such that $|N(b_j)\cap X|\ge k+1$, then $|X|>h$ .
    \end{itemize}
\end{definition}

At the end, we analyze the construction in \cite{KL21} and prove that it is a strong threshold graph. Formally, we have that:

\begin{theorem}
\label{thm:graph}
Given an error correcting code $\mathcal C: \Sigma^r \to \Sigma^m$ with distance $\delta$,
then for any $k\in \mathbb N$ and $\varepsilon \in (0,1)$, a $(|\Sigma|^r,k,|\Sigma|^k,m,\sqrt{\frac{2\cdot \varepsilon}{1-\delta}},\varepsilon)$-threshold graph $G$ can be constructed in time $O((k+m)\cdot |\Sigma|^{O(r+k)})$.
\end{theorem}
\begin{proof}
    The construction is as follows.
    \begin{itemize}
    	\item For every $i \in [k]$,  $A_i=\{\mathcal  C(x):x\in\Sigma^r\}$.
    	\item For every $j \in [m]$,  $B_j=\Sigma^k$.
    	\item A vertex $a \in A_i$ and a vertex $b \in B_j$ are linked if and only if $(a)_j=(b)_i$, where $(a)_j$ means the $j$-th element of vector $a$.
    \end{itemize}
    Fix any $a_1\in A_1, a_2\in A_2, \ldots, a_k\in A_k$ and an index $j\in [m]$. Let $b = ((a_1)_j, (a_2)_j,\ldots,(a_k)_j)$.
    It's easy to see $b \in B_j$ is a common neighbor of $a_1, a_2, \ldots, a_k$.
    
    Fix any $b_1\in B_1, b_2\in B_2, \ldots, b_m\in B_m$ and $X\subseteq A$. Suppose there are $\varepsilon m$ indices $j$ such that $|N(b_j)\cap X|\ge k+1$. Since $A$ is divided into $k$ parts, for each such index $j$, there must be an index $i\in [k]$ such that $N(b_j)$ contains at least two vertices in $X\cap A_i$. 
    Let $x,x'\in X\cap A_i$ be two different vertices, then $(x)_j = (x')_j = (b_j)_i$. 
    Define \[L_{x, x'} = \{i\in [m]: (x)_i = (x')_i\},\] we have \[\sum_{x\neq x'\in X} |L_{x, x'}|\ge \varepsilon m.\]
    However, according to the distance of codewords, for every $x,x' \in X, x \ne x'$, \[|L_{x, x'}|\le (1-\delta)m.\] This leads to $\binom{|X|}{2} (1 - \delta)\ge \varepsilon$, i.e., $|X|> \sqrt{\frac{2\varepsilon}{1-\delta}}$.
\end{proof}

\section{Proof of the Main Theorem}
\label{sec:reduction}

\begin{theorem}
\label{thm:reduction}
There is a reduction which, given an unweighted \textsc{$k$-SetCover} instance $\Gamma=(S,U,E)$, an $(n,k,t,m,h,\varepsilon)$-threshold graph $T = (A\dot\cup B, E_T)$ where $n = |S|$ and $m \le n^{O(1)}$, and an integer $c\in \mathbb{N}^+$, outputs a new $2$-weighted \textsc{$k$-SetCover} instance $\Gamma'=(S',U', E',w)$  with the following properties:
	\begin{itemize}
		\item (Completeness) If $\mathsf{OPT}(\Gamma)\le k$, then $\mathsf{OPT}(\Gamma') \le 2m$.
		\item (Soundness) If $\mathsf{OPT}(\Gamma)>k$, then $\mathsf{OPT}(\Gamma')> \min\{mh/k,(1-\varepsilon)mc\}$.
		\item The reduction runs in $|\Gamma|^{O(1)}\cdot (|U|t)^{O(c)}$ time.
	\end{itemize}
\end{theorem}

\begin{proof}
    Let $A=A_1\dot\cup A_2\dot\cup \ldots\dot\cup A_k$ and $B = B_1\dot\cup B_2\dot\cup \ldots\dot\cup B_m$. 
    For every $i \in [k]$, we treat each $A_i$ as a copy of $[n]$. Let $s: [n]\to S$ be a bijection. 
    The new instance $\Gamma'=(S',U',E')$ and $w$ is defined as follows.
    \begin{itemize}
        \item $S'=A\dot\cup B$.
        \item $U'=\{(u_1,\ldots,u_c,b_1,\ldots,b_c,i): (u_1,\ldots,u_c)\in U^c, (b_1,\ldots,b_c)\in B_i^c,i\in[m]\}.$
        \item For every $a\in A$, $w(a) = m/k$. For every $b\in B$, $w(b) = 1$.
        \item For every $a\in A$ and $\vec{u}=(u_1,\ldots,u_c,b_1,\ldots,b_c,i)\in U'$, we link $a$ and $\vec{u}$ if there exists $j\in[c]$ such that $(a,b_j)\in E_T$ and $(s(a), u_j)\in E$, where $s(a)$ is the matching vertex of $a$ in $S$.
        \item For every $b\in B$ and
		$\vec{u}=(u_1,\ldots,u_c,b_1,\ldots,b_c,i)\in U'$, we link $b$ and $\vec{u}$ if $b\in B_i$ and $b\neq b_j$ for all $j\in[c]$.
    \end{itemize}
    
    It is easy to see the reduction can be done in  
    $\left((nk + tm)\cdot (|U| t)^c\cdot m\right)^{O(1)} = |\Gamma|^{O(1)}\cdot (|U|t)^{O(c)}$ time.
    
For the completeness case, let the solution in $\Gamma$ be $a_{1},\ldots,a_{k}\in S$.
By the property of threshold graph, $a_1\in A_1,\ldots,a_k\in A_k$ have common neighbors $b_1 \in B_1, \ldots, b_m \in B_m$ in $T$. 

We claim that $\{a_1,\ldots,a_k,b_1,\ldots,b_m\} \subseteq  S'$, which is of weight $2m$, is a valid solution of $\Gamma'$. 

For every $\vec{u}=(u_1,\ldots,u_c,\hat b_1,\ldots,\hat b_c,i)\in U'$ with $(\hat b_1,\ldots,\hat b_c) \in B_i^c$ for some $i\in[m]$,

\begin{itemize}
	\item If $b_i$ does not appear in $\{\hat b_1,\ldots,\hat b_c\}$, then according to our construction, $(b_i, \vec{u})\in E'$.
	\item Otherwise suppose $b_i=\hat b_j$ for some $j \in [c]$.  Since $\{a_1,\ldots,a_k\}$ is a valid covering of $U$ in $\Gamma$, there exists $j^*\in[k]$ such that $(s(a_{j^*}), u_j)\in E$. Note that $(a_{j^*},\hat{b}_j)\in E_T$,  we have $(a_{j^*}, \vec{u})\in E'$ by definition.
\end{itemize}

For the soundness case, consider a solution $\{{a_1},\ldots,{a_q},{b_1},\ldots,{b_r}\}$ with $a_i\in A$ and $b_j\in B$, we first prove that, for every $i \in [m]$,
\begin{itemize}
	\item Either $\{b_1,\ldots,b_r\}$  contains at least $c+1$ vertices in $B_i$,
	\item Or there exists $j\in[r]$, such that $b_j\in B_i$ is a common neighbor of at least $k+1$ vertices in $\{a_1,\ldots,a_q\}$ in the threshold graph $T$.
\end{itemize}

Suppose it is not the first case, i.e., the set $\{b_1,\ldots,b_r\}$  contains no more than $c$ vertices in $B_i$, for simplicity of notation let those vertices be $\hat b_1,\ldots, \hat b_c$ (we allow duplication so that the number of vertices can always be $c$).

Consider the set $U_{\hat b_1,\ldots,\hat b_c}=\{(u_1,\ldots,u_c,\hat b_1,\ldots,\hat b_c,i)\in U':(u_1,\ldots,u_c)\in U^c\}$. According to our assumption, $\{{b_1},\ldots,{b_r}\}$ does not cover $U_{\hat b_1,\ldots,\hat b_c}$. So $U_{\hat b_1,\ldots,\hat b_c}$ can only be covered by $\{{a_1},\ldots,{a_q}\}$. Suppose by contradiction that for every $j\in [c]$,  $\hat b_j$ has at most $k$ neighbors in $\{a_1,\ldots,a_q\}$, i.e. the set $N_{\hat b_j}=\{{a_\ell}:\ell\in [q], a_\ell\in N(\hat b_j) \}$ has size at most $k$. Since $\mathsf{OPT}(\Gamma) > k$, there must be some $\hat u_j\in U$ not covered by $N_{\hat b_j}$. Hence for every $j\in [c]$ and $\ell\in [q]$, either $(a_\ell, \hat b_j)\notin E_T$ or $(s(a_{\ell}), \hat u_j)\notin E$, i.e. the vertex $\vec{u}=(\hat u_1,\ldots,\hat u_c,\hat b_1,\ldots,\hat b_c,i)$ is not covered by $\{{a_1},\ldots,{a_q}\}$. 

If there are $\varepsilon m$ indices $i\in[m]$ such that in the threshold graph $T$, there exists $b_j\in B_i$ that has $k+1$ neighbors in $\{a_1,\ldots,a_q\}$, then by the property of threshold graph, $q>h$. It follows that $w(\{a_1,\ldots,a_q\})>h\cdot m/k$.
Otherwise, there are $(1-\varepsilon)m$ indices $i\in[m]$ such that $\{b_1,\ldots,b_r\}$ contains $c+1$ vertices in $B_i$, we have $w(\{b_1,\ldots,b_r\}) > (1-\varepsilon)cm$.
\end{proof}

Now we are ready to prove the $\mathsf{W[2]}$-hardness of constant gap \ksetcover{}.

\begin{theorem}[Restated version of Theorem \ref{w2hardness}]
    Assuming $\mathsf{W[2]} \neq \mathsf{FPT}$, there is no deterministic FPT algorithm which can approximate {\sc\ksetcover{}} within any constant ratio.
\end{theorem}
\begin{proof}
    For any constant $c_0>0$, we give an FPT reduction from a \ksetcover{} instance $\Gamma=(S,U,E)$ to a $c_0$-gap $k'$-\setcover{} instance as follows. 
    Note that without loss of generality, we could assume that our reduction only holds for large enough $\kappa(x)$ and $|x|$, since we could apply enumerative search to solve small instances.
    We assume that $k\ge 2c_0,|S|\ge k^{5k}$ and $|U| \le |S|$. Let
    \begin{itemize}
        \item $n=|S|$,
        \item $r=k$,
        \item $m=k^5$,
        \item $c=4c_0$,
        \item $|\Sigma|=n^{1/k} \ge m$,
        \item $\varepsilon=1/2$.
    \end{itemize}
    By Lemma \ref{lemma:RS}, the Reed-Solomon code $\mathcal{C}^{\rm RS}:\Sigma^r \to \Sigma^m$ has distance at least $\delta=1 - \frac{r}{m}=1-\frac{1}{k^4}$. Then by Theorem \ref{thm:graph}, we can construct a $(|\Sigma|^r=n, k, |\Sigma|^k=n, m=k^5, \sqrt{\frac{2\cdot \varepsilon}{1-\delta}}=k^2, \varepsilon=\frac{1}{2})$-threshold graph in $O((k+m)\cdot |\Sigma|^{O(r+k)})=n^{O(1)}$ time. After that, we apply Theorem \ref{thm:reduction} to get a $2$-weighted gap \textsc{SetCover} instance in time $(n^{O(1)})\cdot (n^2)^{O(c)}=n^{O(c)}$, where in the yes case, the optimal solution has size at most $2m$, and in the no case, the optimal solution has size greater than $\min\{m\cdot k^2/k,(1-\varepsilon)\cdot mc\}=2c_0m$. Finally, we apply Lemma \ref{lem:remove-weights} to remove the weights in $\left(n^{O(c)}\cdot \frac{m}{k}\right)^{O(1)}=n^{O(c)}$ time. The whole reduction runs in time polynomial in $n$, and the new parameter $2m=2k^5$ is a function of $k$. Thus our reduction is an FPT reduction.
\end{proof}

Next we prove the hardness of approximating non-parameterized \ksetcover{} without using the PCP theorem.
\begin{theorem}[Restated version of Theorem \ref{w1hardness}]
    Assuming $\mathsf{W[1]} \neq \mathsf{FPT}$, there is no polynomial time algorithm which can approximate non-parameterized {\sc $k$-SetCover} within $o\left(\frac{\log n}{\log\log n}\right)$ ratio, even if $k$ is as small as $O\left(\frac{\log n}{\log \log n}\right)^3$.
\end{theorem}

\begin{proof}
    Suppose by contradiction there is an algorithm which can approximate non-parameterized \ksetcover{} within a ratio less than $\frac{\log n}{g(n)\log \log n}$ for some $g(n)=\omega(1)$ in $n^{O(1)}$ time, then we give an algorithm which can solve $k$-{\sc Clique} in FPT time. 
    W.l.o.g. assume $g$ is a non-decreasing function.
    
    Given a $k$-{\sc Clique} instance $G = (V, E)$, we first reduce it to a $k'$-\setcover{} instance $\Gamma = (S, U, E)$ where $|U| = O((k')^2\log |S|)$ by Lemma \ref{lemma:smallu}. We will use $k$ for $k'$ in the following for clarity. W.l.o.g. assume $n=|S|$ is large enough that $$\min\left\{g(n), \frac{\log n}{\log \log n}\right\}\ge \max\{2k,100\},$$
    otherwise we use brute-force to solve $\Gamma$, and thus solve $G$ in FPT time. Let 
    \begin{itemize}
        \item $r=\frac{\log n}{\log \log n}$,
        \item $m=\left(\frac{\log n}{\log \log n}\right)^3$,
        \item $c=\frac{\log n}{k \log \log n}$,
        \item $|\Sigma|=\left(\frac{\log n}{\log \log n}\right)^3$,
        \item $\varepsilon=1/2$.
    \end{itemize}
    By Lemma \ref{lemma:RS}, the Reed-Solomon code $\mathcal C^{RS}:\Sigma^r \to \Sigma^m$ has distance at least $\delta=1-\frac{r}{m}=1-\left(\frac{\log \log n}{\log n}\right)^2$. Then by Theorem \ref{thm:graph}, we can construct a $(|\Sigma|^r=n^{O(1)},k,|\Sigma|^k = \left(\frac{\log n}{\log \log n}\right)^{3k},m=k^5, \sqrt{\frac{2\cdot \varepsilon}{1-\delta}}=\frac{\log n}{\log \log n},\varepsilon=\frac{1}{2})$-threshold graph in $O((k+m)\cdot |\Sigma|^{O(r+k)})=n^{O(1)}$ time. After that, we apply Theorem \ref{thm:reduction} to get a 2-weighted gap \setcover{} instance in time $n^{O(1)}\cdot \left(k^2 \log n\cdot \left(\frac{\log n}{\log \log n}\right)^{3k}\right)^{c}=n^{O(1)}$, where in the yes case, the optimal solution has size at most $2m=2\left(\frac{\log n}{\log \log n}\right)^3$ and in the no case, the optimal solution has size greater than $\min\{m\cdot \frac{\log n}{\log \log n}/k, (1-\varepsilon)\cdot mc\}=\frac{\log n}{2k \log \log n}$. We apply Lemma \ref{lem:remove-weights} to remove the weights in $n^{O(1)}$ time. Finally, as $\frac{\log n}{2k \log \log n}\ge \frac{\log n}{g(n) \log \log n}$, the presumed approximation algorithm for non-parameterized \ksetcover{} can distinguish between the two cases in $n^{O(1)}$ time, and thus solve {\sc $k$-Clique} in FPT time, contradicting $\wone\neq\fpt$.
\end{proof}
\section{Conclusion}
\label{sec:conclusion}
In this paper, we settle the $\wtwo$-hardness of approximating \ksetcover{} with constant approximation ratio. Our result could also be applied to rule out polynomial algorithm approximating non-parameterized {\sc $k$-SetCover} within ratio $o\left(\frac{\log n}{\log \log n}\right)$, with $k$ as small as $O\left(\frac{\log n}{\log \log n}\right)^3$, assuming $\wone \ne \fpt$.  

As further research questions, it is interesting to consider $\wtwo$-hardness of approximating \ksetcover{} beyond constant ratio.

\begin{question}
Is it $\wtwo$-hard to approximate \ksetcover{} with super-constant ratio?
\end{question}

Our technique fails to answer this question, because given approximation ratio $c$, our reduction runs in time $\Omega(|U|^{c})$, which will result in a non-FPT reduction if $c$ is not constant.

It is well-known that the textbook greedy algorithm approximates \ksetcover{} with ratio $O(\log n)$, and this is also the best approximation algorithm. However, assuming $\wone\ne\fpt$, the state-of-the-art result \cite{Lin19} rules out FPT algorithms approximating \ksetcover{} within approximation ratio $\left(\frac{\log n}{\log \log n}\right)^{1/k}$, and there is still a gap. Could we further improve this inapproximability, or does there exist FPT algorithms approximating \ksetcover{} with better approximation ratio? This leads to the following question:

\begin{question}
Is there any  FPT algorithm approximating \ksetcover{} within approximation ratio $o(\log n)$?
\end{question} 


Finally, we ask whether the lower bound of our Theorem \ref{w1hardness} could be improved under stronger assumptions (e.g., $\mathsf{SETH}$), which leads to the following question:

\begin{question}
For $k = (\log n)^{O(1)}$, assuming $\mathsf{SETH}$, is there any algorithm with running time $n^{k-o(1)}$ approximating \ksetcover{} with approximation ratio $o\left(\frac{\log n}{\log \log n}\right)$?
\end{question}




\bibliographystyle{alpha}
\bibliography{main}
\section*{Appendix}
\label{appendix:A}
We prove that, for any computable function $f$,  \ksetcover{} with universe set size at most $f(k)\log n$ is in $\mathsf{W[1]}$. We show a reduction from such a \ksetcover{} instance to a {\sc $k'$-Clique} instance. The idea is as follows. For each element in $U$, it should be covered by some set from one of the $k$ groups. We divide $U$ into $f(k) \cdot \log k$ parts of size $\log n/\log k$, and use a variable in $[k]^{\log n/\log k}$ for each part to encode which groups the sets covering those $\log n/\log k$ elements are from. For each group of sets, we also use a variable to indicate which set is picked from this group. We add constraints between the two types of variables to check whether the alleged set is eligible for covering corresponding elements, and transform this 2CSP instance to a {\sc Clique} instance in the canonical way.
\begin{lemma}\label{lem:setcoverw1}
There is an FPT-reduction which, given an instance $\Gamma=(S,U,E)$ of \ksetcover{} with $|U|=f(k)\cdot \log n$, outputs a graph $G$ and an integer $k'=k+f(k)\cdot\log k$, such that
\begin{itemize}
    \item if $\mathsf{OPT}(\Gamma)=k$, then $G$ contains a $k'$-clique,
    \item  if $\mathsf{OPT}(\Gamma)<k$, then $G$ contains no $k'$-clique.
\end{itemize}
\end{lemma}
\begin{proof}
Let $h=f(k)\cdot \log k$. We divide $U$ into $h$ groups of size $\log n/\log k$ and index every element in $U$ by a pair $(i,j)$ where $i \in [h],j \in [\log n/\log k]$. We define a graph $G$ as follows.
\begin{itemize}
    \item $V(G)=V_1\cup V_2,\cdots,V_k\cup W_1\cup W_2,\cdots,W_h$. 
    \item For every $i \in [k]$, $V_i$ is a copy of $S$. 
    \item For every $i\in[h]$, $W_i=[k]^{\log n/\log k}$. 
    \item Make each $V_i$ and $W_j$ an independent set. Add edges between different $V_i$ and $V_j$, and between different $W_i$ and $W_j$. 
    \item For every $v\in V_i$ and $w\in W_j$, add an edge between them if and only if for all $\ell \in [\log n/\log k]$, $w[\ell]=i$ implies $v$ can cover the $(j,\ell)$-th element in $U$.
\end{itemize}
The running time of this reduction is at most $O(kh\cdot k^{\log n/\log k}\cdot |G|\cdot |U|)=k\log k\cdot f(k)\cdot n^{O(1)}$.

Suppose $\Gamma$ has a size-$k$ solution $v_1,\ldots,v_k\in S$. We define $w_1\in W_1,\ldots,w_h\in W_h$ as follows. For every $j\in[h]$ and $\ell\in[\log n/\log k]$, let $w_j[\ell]=i$ such that $v_i$ can cover the $(j,\ell)$-th element in $U$. By our construction, each $w_j$ is adjacent to all  $v_1,\ldots,v_k$. Thus, we obtain a $k'$-clique in graph $G$.

Suppose $G$ contains a clique $X$ of size $k'$. By our construction, $|X \cap V_i|=|X \cap W_j|=1$ for every $i\in[k],j \in [h]$. Let $v_i\in |X \cap V_i|,w_j\in |X \cap W_j|$ be the vertices in the clique. For every $j\in[h]$ and $\ell\in[\log n/\log k]$, let $i=w_j[\ell]$, then by the existence of an edge between $v_i$ and $w_j$, the $(j,\ell)$-th element in $U$ can be covered by $v_i$. Thus, $U$ can be covered by $k$ sets $v_1,\ldots,v_k$.
\end{proof}
\end{document}